\documentclass[pra,onecolumn,superscriptaddress]{revtex4}

\usepackage{mathtools}
\usepackage{float}
\usepackage{graphicx}
\usepackage{epsfig}
\usepackage{dcolumn}
\usepackage{amsmath,amssymb,amsthm}
\usepackage{bm}
\usepackage{verbatim}
\usepackage{natbib}
	
\usepackage{color}

\newcommand{\bra}[1]{\left\langle #1 \right|}
\newcommand{\ket}[1]{\left| #1 \right\rangle}
\newcommand{\braket}[2]{\left\langle #1 \middle| #2 \right\rangle}
\newcommand{\ketbra}[2]{\left|#1\middle\rangle\middle\langle#2\right|}

\newcommand{\id}{\mathbb{I}}
\newcommand{\M}[1]{\mathcal{#1}}
\newcommand{\B}[1]{\mathbb{#1}}
\newcommand{\Tr}{\text{Tr}}

\newcommand{\apos}{\textsc{\char13}}

\newtheorem{teo}{Theorem}

\newtheorem{defi}[teo]{Definition}
\newtheorem{cor}[teo]{Corolary}
\newtheorem{prop}[teo]{Proposition}

\begin{document}
\title{Role of Quantumness of Correlations in Entanglement Resource Theory} 
\author{Tiago Debarba} 
\affiliation{Universidade Tecnol\'ogica Federal do Paran\'a (UTFPR), Campus Corn\'elio Proc\'opio, Avenida Alberto Carazzai 1640, Corn\'elio Proc\'opio, Paran\'a 86300-000, Brazil}
\email{debarba@utfpr.edu.br}

\date{\today}
\begin{abstract}
Quantum correlations: entanglement and quantumness of correlations are main resource for quantum information theory. In this chapter it is presented the scenarios which    quantumness of correlations plays an interesting role in entanglement distillation protocol.  By means of Koashi - Winter relation, it is discussed that quantumness of correlations are related to the irreversibility of the entanglement distillation protocol. The activation protocol is introduced, and it is proved that quantumness of correlations can create distillable entanglement between the system and the measurement apparatus during a local measurement process.  
\end{abstract}


\maketitle
\section{Introduction}
Quantum entanglement plays the fundamental role in quantum information and computation 
\cite{chuang,horodeckireview}. The resource theory of quantum entanglement, entanglement distillation 
\cite{devetak05} and entanglement cost \cite{hayden01}, 
revealed one of most fundamental aspects of quantum mechanics.  
Entanglement distillation protocol consist into convert a number of copies of an entangled state in few 
copies of maximally entangled states, by means of local operations and classical communication (LOCC) 
\cite{bennettprl}. As maximally entangled states are the main resource of the quantum information,  
entanglement distillation protocol has many applications in this scenario, as quantum teleport \cite{teleport}, 
 quantum error correction \cite{bennettpra} and quantum cryptography \cite{bb84}. A family of quantum 
 information protocol arises from distillation of quantum entanglement and secret keys \cite{mother,devetak05}

 However independently Ollivier and Zurek \cite{zurek2001}, and 
Henderson and Vedral \cite{henderson} found a 
new quantum property, without counterpart in classical systems.   
They named it as the  {\it quantumness of correlations}. 
 This new kind of correlation reveals the amount of 
information destroyed during the local 
measurement process, and goes beyond the quantum entanglement. 
There are many equivalent formulations for characterization and quantification of quantumness 
of correlations: quantum discord \cite{zurek2001,henderson}, minimum local disturbance 
\cite{opp,fu08,taka} and geometrical  approach \cite{luofu,modiprl2010,debarbapra}. 

This chapter presents in details two different ways to relate quantum entanglement 
and quantumness of correlations. The main purpose of this chapter is to discuss that 
quantumness of correlations plays an interesting role in entanglement distillation protocol.  
Entanglement and quantumness of correlations connect each other in two different pictures. 
The relation derived by Koashi and Winter \cite{kw04} demonstrates the balance between quantumness of 
correlations and entanglement in the purification process \cite{fanchini11}.  
This balance leads to a formal proof for the irreversibility of the entanglement distillation protocol, 
in terms of quantumness of correlations \cite{cornelio11}. 
 In the named {\it activation protocol}, the quantumness of correlations of a given composed 
 system can be converted into distillable entanglement with a measurement apparatus during the local measurement process \cite{pianiprl, streltsov11}.   
 
The chapter is organized as follow. 
In Sec.\ref{sec.math} a mathematical overview is presented, 
and the notation is defined. Sec.\ref{q.c.} introduces  
some important concepts about the notion of quantum correlations: entanglement and 
quantumness of correlations. 
In Sec.\ref{seckw} is presented the Koashi-Winter relation, and its role in the irreversibility of 
quantum distillation process. 
Section \ref{secap} is intended to 
the description of the activation protocol, and the demonstration that 
quantumness of correlation can be activated  into distillable entanglement.  

\section{Mathematical Overview}\label{sec.math}
This section introduces some quantum information concepts and defines the notation used in the chapter. 
\subsection{Density Matrix and Quantum Channels}
As the convex combination of positive matrices is also positive, then the space of positive operators forms a convex cone in Hilbert - Schmidt $\M{L}(\B{C}^N)$ \cite{zyc}. 
If we restrict the matrices in the positive cone to be trace=$1$, we arrive to another set of matrices, that is named 
the set of density matrices. This set of operators also originates a vector space, this space is denote as $\M{D}(\B{C}^N)$. Therefore the matrices that belong 
to this set, or the vectors in this vector space, are named {\it density matrices}.
\begin{defi}
A linear positive operator $\rho\in\M{D}(\B{C}^N)$ is a density matrix, and represents the state of 
a quantum system, if it satisfies the following properties:
\begin{itemize}
\item{Hermitian}: $\rho = \rho^{\dagger}$;
\item{Positive semi-definite}: $\rho \geq 0;$
\item{Trace one}: $\Tr(\rho) =1$
\end{itemize} 
\end{defi}
As the convex combination of density matrices is a density matrix, 
the vector space $\M{D}$ is a convex set whose pure states are projectors 
onto the real numbers. A given density matrix $\rho\in\M{D}(\B{C}^N)$ is 
a pure state if it satisfies:
\begin{eqnarray}
\rho &=& \rho^2,
\end{eqnarray}
then the state $\rho$ is a rank-1 matrix and it can be written as:
\begin{equation}
\rho = \ketbra{\psi}{\psi}. 
\end{equation}
The set of pure states is  a $2(N-1)$-dimensional subset of the $(N^2-2)$-dimensional boundary of 
$\M{D}(\B{C}^N)$. Every state with at least one eigenvalue equal to zero 
belongs to the boundary \cite{zyc}. For 2-dimensional systems 
(it is also named qubit \cite{schumacher95}) the boundary is just composed by pure states. 

Consider a linear transformation $\Phi: \M{L}(\B{C}^N)\rightarrow \M{L}(\B{C}^M)$. 
This map represents a physical process, if it satisfies some conditions, 
determined by the mathematical properties of the density matrices.
Indeed, to represent a physical process the transformation must map a quantum 
state into another quantum state, $\Phi:\M{D}(\B{C}^N)\rightarrow\M{D}(\B{C}^M)$. 
It holds if $\Phi$ satisfy the following properties:
\begin{itemize}
\item{\bf Linearity:}
As a quantum state can be a convex combination of other quantum states, 
the map must be linear. For two arbitrary operators $\rho,\sigma\in\M{D}(\B{C}^N)$
\begin{equation}
\Phi(\rho + \sigma) = \Phi(\rho) + \Phi(\sigma); 
\end{equation}

\item{\bf Trace preserving:}
The eigenvalues of the density matrix represent probabilities, and it sum must be one, then a quantum channel must to keep the trace of the density matrix:
\begin{equation}
\Tr[\Phi(\rho)]=1. 
\end{equation}
\item{\bf Completely Positive:} Consider a channel $\Phi:\M{D}(\B{C}_A)\rightarrow\M{D}(\B{C}_A)$ and a quantum state  $\rho,\sigma\in\M{D}(\B{C}_A\otimes\B{C}_B)$, then
\begin{equation}
\id\otimes\Phi(\rho)\geq 0.
\end{equation}
\end{itemize} 
The map that satisfies this property is named {\it completely positive map}.
The linear transformations which map quantum states in quantum states 
are named Completely Positive and Trace Preserve (CPTP) {\it quantum channels}. 
The space of quantum channels that maps $N\times N$ density matrices onto $M\times M$ density matrices 
is denoted as $\M{C}(\B{C}^N,\B{C}^M)$.

\subsection{Measurement}
\label{chap1.2.4}

Measurement is a classical statistical inference of quantum systems. 
The measurement process maps a quantum state into a classical probability distribution.

We can define a measurement as a function $\Pi: \Sigma \rightarrow \M{P}(\B{C}_{\Gamma})$ \footnote{
Just to clarify the notation, when we write a subscript  in the complex 
euclidean vector space, as $\B{C}_{\Gamma}$, it represents a label to the space, it shall be 
very useful when we study composed systems. When we write a superscript on it, it  
represents the dimension of the complex vector space. For example, if $\text{dim}(\B{C}_{\Gamma}) = N$, 
we can also represent this space as $\B{C}^N$, the usage of the notation will depend on the context.}, 
associating 
an alphabet $\Sigma$ to positive operators $\{\Pi_x \}_x\subset\M{P}(\B{C}_{\Gamma})$.
For a given density matrix $\rho\in\M{D}(\B{C}_{\Gamma})$, the measurement process consist in to chose an element of $\Sigma$ randomly. 
This random choice is represented by a probability vector $\vec{p}\in \B{R}^{N}_{+}$, with $N$ being 
the cardinality of the random variable described by $\vec{p}$. 
The elements of the probability vector $\vec{p}$ are given by:
\begin{equation}
p_x = \Tr(\Pi_x \rho),
\end{equation}
where $\Pi_x$ is the {\it measurement operator} associated to $x\in\Sigma$. The alphabet $\Sigma$ 
is the set of measurement outcomes, and the vector $\vec{p}$ is the classical probability vector 
associated to the measurement process $\Pi$ of a given density matrix $\rho$. 
As the outcomes are elements of a probability vector, these elements must be positive and sum to one. 
Which implies that the measurement operators must sum to identity: 
\begin{equation}
\sum_x \Pi_x = \id_{\Gamma},  
\end{equation} 
where $\id_{\Gamma}$ is the identity matrix in $\B{C}_{\Gamma}$. It is easy to check that 
this condition implies $\sum_{x}p_x=1$:
\begin{equation}
\sum_x p_x  = \sum_x \Tr(\Pi_x \rho) = \Tr(\sum_x \Pi_x \rho) = \Tr(\rho) = 1.
\end{equation}
For instance we shall restrict the measurements to a subclass of measurement operators named 
{\it projective measurements}. As it is shown later, its generalization can be performed via the Naimark's theorem. 
For projective measurements the cardinality of 
$\vec{p}$ is at least the dimension of $\rho$, and the measurement operators are projectors:
\begin{equation}
\Pi_x^2 = \Pi_x, 
\end{equation} 
for any $x\in\Sigma$. 
If we consider an orthonormal basis $\{ \ket{e_x} \}$, where the vectors $\ket{e_x}$ span 
$\B{C}_{\Gamma}$, this set represents a projective measurement for $\Pi_x = \ketbra{e_x}{e_x}$. 
The output state is described by the expression:
\begin{equation}
\rho_x = \frac{\Pi_x \rho \Pi_x}{\Tr(\Pi_x \rho)}. 
\end{equation}
The set of operators defines a convex hull in  $\M{P}(\B{C}_{\Gamma})$, then a measured state 
represents a pure state in this convex hull. 
In this way, the post - measurement state can be reconstruct by the convex  combination of the output states $\rho = \sum_x p_x \rho_x$.

As physical processes are described by quantum channels, it is possible to describe the classical 
statistical inference of the quantum measurements as a CPTP channel.
A channel that maps a quantum state into a probability vector is the dephasing channel. Therefore the post - measurement state is 
the state under the action of the dephasing channel. . 
\begin{teo}
A given map $\Phi\in\M{C}(\B{C}_{\Gamma},\B{C}_{\Gamma'})$ is a measurement if and only if:
\begin{equation}\label{g666ff}
\Phi(\rho) = \sum_x \Tr(M_x \rho)\ketbra{e_x}{e_x},
\end{equation}
where $\rho\in\M{D}(\B{C}_{\Gamma})$, $M_x\in\M{P}(\B{C}_{\Gamma})$ and $\ket{e_x}\in\B{C}_{\Gamma'}$.
\end{teo} 
In order to differ the set of measurement channels from a general CPTP channel, 
this set is represented as $\M{P}$. 
A given measurement map 
$\M{M}\in\M{P}(\B{C}_{\Gamma},\B{C}_{\Gamma'})$ is a quantum channel that maps 
a density matrix in a probability vector,   
$\M{M}:\M{D}(\B{C}_{\Gamma})\rightarrow \B{R}_{\Gamma'}^{+}$. This probability 
vector is described by a diagonal density matrix as in Eq.\eqref{g666ff}.
The dimension of $\B{C}_{\Gamma'}$ is the number of outcomes of the measurement.

For general measurements, described by Positive Operators Valued Measure (POVM), the measurement process can be 
described by a measurement channel $\Phi\in\M{P}(\B{C}_{\Gamma},\B{C}_{\Gamma’})$. 
The description performed above can be follow to describe these general measurements, indeed projective measurements are a 
a restriction for a POVM composed by orthogonal operators. 
Consider a set of positive operators $\{ M_x \}_x$,  representing a POVM, then  
$\Tr[M_x\rho] = p_x$ are the elements of a probability vector $\vec{p}\in\B{R}_{\Gamma’}^{+}$, then the post - measurement state is:
\begin{equation}
\Phi(\rho) = \sum_x p_x \ketbra{e_x}{e_x}. 
\end{equation}
where $\{ \ket{e_x} \}_x$ is an orthonormal basis in $\B{C}_{\Gamma’}$.

Via the Naimark's theorem it is possible the measurement channel is described as a dephasing 
channel on a state in a enlarged space.   
In other words, for POVMs whose elements are 
rank-1 and linearly independent, 
it is possible to associate a projective measurement on an enlarged space. 
\begin{teo}[Naimark's Theorem]
Given a quantum measurement $\M{M}\in\M{P}(\B{C}_{\Gamma},\B{C}_{\Gamma'})$, 
with POVM elements $\{M_x \}_{x=0}^{M}$,   
there exists a projective measurement $\Pi\in\M{P}(\B{C}_{\Gamma'})$, 
with elements $\{ \Pi_y\}_{y=0}^{M}$ such that:
\begin{equation}
\Tr(M_x\rho) = \Tr(\Pi_x V \rho V^{\dagger}), 
\end{equation}
where $V \in \M{U}(\B{C}_{\Gamma},\B{C}_{\Gamma'})$ is an isometry.
\end{teo}
The action of the isometry on the state $\rho$, in the Naimark's theorem, is named 
embedding operation. The simplest way to embed the state $\rho$, it is simply coupling a pure 
ancillary system. In this way the isometry will be $V = \id_{\Gamma}\otimes\ket{0}_E$ 
and the enlarged space $\B{C}_{\Gamma'}=\B{C}_{\Gamma}\otimes\B{C}_{E}$. For this 
simple case the relation between the POVM elements $\{M_x\}_x$ and the projective measurement 
on the enlarged space  $\{ \Pi_x \}_x$:
\begin{equation}
M_x = (\id_{\Gamma}\otimes\bra{0}_E) \Pi_x (\id_{\Gamma}\otimes\ket{0}_E). 
\end{equation}

As the measurement can be described by a quantum channel, we can study how quantum measurements 
can be performed locally. 
\begin{defi}
Given a $N$-partite composed system, represented by the state 
$\rho_{A_1,...,A_N}\in\M{D}(\B{C}_{A_1}\otimes\cdots\otimes\B{C}_{A_N})$, 
we define the measurement on each subsystem applied locally: 
\begin{equation}
\Phi_{A_1}\otimes\cdots\otimes\Phi_{A_N}(\rho_{A_1,...,A_N}) = \sum_{\vec{k}}
\Tr[M^{A_1}_{k_1}\otimes\cdots\otimes M^{A_N}_{k_N}\rho_{A_1,...,A_N}]\ketbra{\vec{k}}{\vec{k}},
\end{equation}
where $\ket{\vec{k}} = \ket{k_1}\otimes\cdots\otimes\ket{k_N}$ and the label $\vec{k}$ in the sum  
represents the set of indexes $k_1,...,k_N$. $\{ M_{k_x}^{A_x}\}_{k_x}$ 
are the measurement operators on each subsystem.
\end{defi}

Suppose the measurement is performed on some subsystems, remaining other subsystems unmeasured.  
Consider a 
bipartite system $\rho_{AB}\in\M{D}(\B{C}_A\otimes\B{C}_B)$, and a measurement 
acting on the system $B$, then the measurement map will be written as:
\begin{equation}
\id_A\otimes\Phi_B(\rho_{AB}) = \sum_x\Tr_B[\id_A\otimes M^B_x \rho_{AB}]\otimes \ketbra{b_x}{b_x}.
\end{equation} 
As the measurement is not acting on $A$, the post-measured state on $A$ will remain the same. 
If we write $p_x = \Tr_{AB}[\id_A\otimes M^B_x \rho_{AB}]$ and $\rho_x^{A} = 
\frac{\Tr_B[\id_A\otimes M^B_x \rho_{AB}]}{\Tr_{AB}[\id_A\otimes M^B_x \rho_{AB}]}$,  
the post-measured state will be:
\begin{equation}
\id_A\otimes\Phi_B(\rho_{AB}) = \sum_x p_x \rho_x^{A} \otimes \ketbra{b_x}{b_x}. 
\end{equation}
As the measurement is a classical statistical inference process, the local measurement process destroys the quantum correlations between the 
systems. Indeed the  post-measured state is not a classical probability distribution, although it 
only has classical correlations.

\subsection{Quantum Entropy}
Consider that one can prepare an ensemble of 
quantum states $\xi = \{p_x,\rho_x \}_x$, accordingly to some random variable $X$. 
Classical information can be extracted from the ensemble of quantum states, in the 
form of a variable $Y$, performing measurements on the quantum system. 
The conditional probability distribution to obtain a value $y$,  
given as input the state $\rho_x$ is:
\begin{equation}
p(y|x) = \Tr(M_y\rho_x), 
\end{equation}
where $\{ M_y\}_y$ is a POVM.  
The joint probability distribution $X$ and $Y$ is given by:
\begin{equation}
p(x,y) = p_x\Tr(M_y\rho_x).
\end{equation}
The probability distribution of $Y$ is obtained from the marginal probability distribution:
\begin{equation}
p(y)=\sum_x p(x,y) = \sum_x p_x\Tr(M_y\rho_x) = \Tr(M_y \sum_x p_x\rho_x). 
\end{equation} 
Considering the Bayes rule:
\begin{equation}
p(x,y) = p_x p(y|x) = p(y) P(x|y), 
\end{equation}
it is possible to obtain the conditional probability distribution with elements:
\begin{equation}
P(x|y)=\frac{p_x p(y|x)}{p(y)}. 
\end{equation}

Even in the case the system is always  prepared in the same state, 
there exists an uncertainty about the measured of an observable. 
The probability distributions presented above are 
evidencing this uncertainty, for the measurement observables of a POVM. 
These probability distributions are classical probability distributions extracted from the quantum 
systems, and the Shannon entropy quantifies the degree of 
surprise related to a given result.

It is also possible to define a quantum analogous to the Shannon 
entropy. This quantum entropy is named {von Neumann  
entropy}, and in analogy with Shannon entropy, it is defined as 
the expectation value of the operator $\log_2(\rho)$. 
\begin{defi}[von Neumann entropy]
Given a density operator $\rho\in\M{D}(\B{C}^N)$, the quantum version of 
the Shannon entropy is defined as the function:
\begin{equation}
S(\rho) = - \Tr[\rho \log_2{\rho}].
\end{equation}
\end{defi}
The von Neumann entropy can be rewritten as:
\begin{equation}
S(\rho) = - \sum_k \lambda_k \log_2(\lambda_k),
\end{equation}
where $\{ \lambda_k\}_k$ are the eigenvalues of $\rho=\sum_k \lambda_k \ketbra{k}{k}$. The von Neumann 
entropy has the same interpretation of the Shannon entropy for the probability 
distribution composed by the eigenvalues of the density matrix.
The von Neumann 
entropy is zero of pure states, and it is maximum for  the  
maximally mixed state $\id/N$, where it is $S(\id/N) = \log_2N$. 

For  composed systems the von Neumann entropy is analogous to the Shannon entropy of the 
joint probability. For a bipartite state $\rho_{AB}$, the joint von Neumann 
entropy is:
\begin{equation}
S(\rho_{AB}) = -\Tr(\rho_{AB} \log_2 \rho_{AB}).
\end{equation}

Follow some interesting, and useful, properties about von Neumann entropy: 
\begin{enumerate}
\item{(\bf Pure States)} \label{teo50} For a bipartite pure state $\ket{\phi}_{AB}\in\B{C}_A\otimes\B{C}_B$, 
the partitions have the same von Neumann entropy:
\begin{equation}
S(\rho_A) = S(\rho_B),
\end{equation}
where $\rho_A = \Tr_B(\ketbra{\phi}{\phi}_{AB})$.

\item{(\bf Additivity)} \label{prop5dd} von Neumann entropy is additive:
\begin{equation}
S(\rho\otimes\sigma) = S(\rho)+ S(\sigma),
\end{equation}
where $\rho$ and $\sigma$ are density matrices.

\item{(\bf Concavity)} \label{item.3} von Neumann entropy is a concave function: 
\begin{equation}
S(\sum_ip_i \rho_i) \geq \sum_ip_i S(\rho_i),
\end{equation}
for a convex combination $\rho=\sum_i p_i \rho_i $.  

\item{(\bf Classical - Quantum states)}\label{item.4} For bipartite state in the form $\rho_{AB} = \sum_x p_x \ketbra{x}{x}\otimes\rho_x$, 
the von Neumann entropy will be:		
\begin{equation}
S(\sum_x p_x \ketbra{x}{x}\otimes\rho_x) =H(X)+ \sum_x p_x S(\rho_x), 
\end{equation}
where $H(X)= - \sum_x p_x \log_2 p_x $ 

\end{enumerate}

For composed system it is possible to define a quantum analogous to the mutual information for bipartite 
states.  
\begin{defi}[Mutual information]
Given a bipartite state $\rho_{AB}\in\M{D}(\B{C}_A\otimes\B{C}_B)$ the 
quantum mutual information is defined as:
\begin{equation}
I(A:B)_{\rho_{AB}} = S(\rho_A) + S(\rho_B) - S(\rho_{AB}).
\end{equation}
\end{defi} 
The quantum mutual information of $\rho_{AB}$ quantifies the correlations 
in quantum systems. It can be interpreted as the number of qubits that  
one part must send to another to destroy the correlations between the entire system. 
As the amount of correlations in a quantum state must be positive, it is possible to conclude that:
\begin{equation}
S(\rho_A) + S(\rho_B) \geq S(\rho_{AB}).
\end{equation}
From property \ref{prop5dd}, it is easy to see that mutual information 
is zero for product state $\rho_{AB} = \rho_A\otimes\rho_B$. 
The mutual information of pure states will be equal to:
\begin{equation}
I(A:B)_{\psi_{AB}} = 2 S(\rho_A) = 2S(\rho_B),
\end{equation}
where $\psi_{AB} = \ketbra{\psi}{\psi}_{AB}$ is pure state. 

The quantum version of the relative entropy quantifies the distinguishability 
between quantum states. 
\begin{defi}[Quantum relative entropy]
Given two density matrices $\rho,\sigma\in\M{D}(\B{C}^N)$, the distinguishability between them can be 
quantified using the quantum relative entropy:
\begin{equation}
S(\rho||\sigma) = \Tr[\rho\log_2{\rho} -\rho\log_2{\sigma}]. 
\end{equation}
It will be zero if $\rho = \sigma$.
\end{defi}
The quantum relative entropy is a positive function for $\text{supp}(\rho)\subseteq \text{supp}(\sigma)$, otherwise 
it diverges to infinity. 
The quantum mutual information also can be written as a quantum relative 
entropy. 
\begin{prop}\label{prop666}
Consider a bipartite state $\rho_{AB}$, the following expression holds:
\begin{equation}
I(A:B)_{\rho_{AB}} = S(\rho_{AB}||\rho_A\otimes\rho_B),
\end{equation}
where $\rho_A$ and $\rho_B$ are the reduced states of $\rho_{AB}$.
\end{prop} 
In contrast with the von Neumman entropy, 
the relative entropy always decreases under the action of a quantum channel. 
This property has an operational meaning: 
two states are always less distinguishable under the 
action of noise. 
\begin{teo}
Given two density matrices $\rho,\sigma\in\M{D}(\B{C}_A)$ and a quantum channel 
$\Gamma\in\M{C}(\B{C}_A,\B{C}_B)$, the following inequality holds:
\begin{equation}
S(\rho||\sigma) \geq S(\Gamma(\rho)||\Gamma(\sigma))
\end{equation}
\end{teo}

This theorem implies into another property of the quantum mutual information:  
it decreases monotonically under local CPTP channels. As mutual information quantifies 
correlations, this means that the amount of 
correlations reduce under local noise. 
\begin{cor}
Given a bipartite state $\rho_{AB}\in\M{D}(\B{C}_A\otimes\B{C}_B)$ 
and quantum channel $\Phi_B\in\M{C}(\B{C}_B,\B{C}_{B'})$, the mutual information 
satisfies:
\begin{equation}
I(A:B)_{\rho_{AB}} \geq I(A:B')_{\id\otimes\Phi(\rho_{AB})}.
\end{equation}
\end{cor}
\begin{proof}
Given the mutual information:
\begin{equation}
I(A:B)_{\rho_{AB}} = S(\rho_{AB}||\rho_A\otimes\rho_B) 
\end{equation}
using the theorem above:
\begin{equation}
I(A:B)_{\rho_{AB}} \geq S(\id_A\otimes\Phi_B(\rho_{AB})||\rho_A\otimes\Phi_B(\rho_B)) = I(A:B')_{\id\otimes\Phi(\rho_{AB})}. 
\end{equation}
\end{proof}

Analogous to the classical conditional  entropy, it is possible to define a quantum version of it. 
For a bipartite system $\rho_{AB}$, the quantum conditional entropy quantifies the amount of information of $A$ that is available when $B$ is known. 
\begin{defi}[Conditional entropy]
Consider a bipartite system $\rho_{AB}$, the quantum conditional entropy is defined 
as the function:
\begin{equation}
S(A|B)_{\rho_{AB}} = S(\rho_{AB}) -S(\rho_B).
\end{equation}
\end{defi} 
One interesting property of the quantum conditional entropy is that 
it can be negative. For example, if we consider a bipartite pure state 
$\ket{\phi}_{AB} = (\ket{00}+\ket{11})/\sqrt{2}$, von Neumann entropy 
of the pure state is zero: $S(\ketbra{\phi}{\phi}_{AB}) =0$. 
Nonetheless the reduced state  is the maximally mixed state: $\rho_B = \id/2$, 
whose von Neumann entropy is $S(\id/2) = 1$. 
Therefore the conditional entropy of this state is negative 
$S(A|B)_{\ketbra{\phi}{\phi}_{AB})} = -1$. 
The negative value of the quantum conditional entropy is defined as the {\it coherent information}:
\begin{equation}
I(A\rangle B) = - S(A|B). 
\end{equation}
The conditional entropy has an operational meaning in the state merging protocol,  
where a tripartite pure state is shared by two experimentalists, one will send part of its 
state through a quantum channel to the other. The coherent information 
quantifies the amount of entanglement required to the sender be able to perform 
the protocol. If it is positive, they cannot use entanglement to perform the state merging, 
and in the end the  amount of entanglement grows \cite{esm,dattasm,sm}. 
The coherent information 
also quantifies the capacity of a quantum channel, optimizing over all input states 
$\rho_A$, the output state is known to be $\rho_B$.  
This result is named as {\it LSD theorem} 
(\cite{devetaklsd,lloyd97,shor02}, apud \cite{wilde}). 

\section{Quantum Correlations}\label{q.c.}
\subsection{Entanglement}
This section introduces the concept of quantum entanglement, presenting its  characterization and quantification. 
\subsubsection{Separable States}
Consider two systems $A$ and $B$, often named the experimentalists responsible 
by the systems as Alice and Bob respectively. The state of the systems $A$ and $B$ is 
described by a density matrix on a Hilbert space. In this way considering two 
finite Hilbert spaces $\B{C}_A$ and $\B{C}_B$, and a basis in each one: 
\begin{eqnarray}
&&\{\ket{a_i} \}_{i=0}^{|A|-1}\in\B{C}_A; \\
&&\{\ket{b_k} \}_{k=0}^{|B|-1}\in\B{C}_B,
\end{eqnarray} 
where $|A| = \text{dim}(\B{C}_A)$ and $|B| = \text{dim}(\B{C}_B)$. The global system, composed 
by $A$ and $B$, can be obtained through the tensor product between the basis in the Hilbert space 
of each system:
\begin{equation}\label{t1}
\{\ket{a_i,b_k}\}_{i,j=0}^{|AB|-1} = \{\ket{a_i}\otimes\ket{b_k}\}_{i,k=0}^{|A|-1,|B|-1}, 
\end{equation} 
hence the dimension of the composed system is the product of the dimension: 
$|AB|=\text{dim}(\B{C}_{AB}) = \text{dim}(\B{C}_{A})\cdot\text{dim}(\B{C}_{B})$. The  Hilbert space 
of the composed system is denoted as 
$\B{C}_{AB} = \B{C}_{A}\otimes\B{C}_{B}$. A pure state of the 
composed system can be decomposed in the basis in Eq.\ref{t1}:
\begin{equation}\label{t2}
\ket{\psi}_{AB} = \sum_{i,k} c_{i,k}\ket{a_i}\otimes\ket{b_k}. 
\end{equation} 
From this expression one can realize that: in general a pure state, that describes a composed 
system, cannot be written as the product of the state of each system. In other words, suppose the 
system $A$ and $B$ described by the states $\ket{\alpha}_A=\sum_i a_i\ket{a_i}\in\B{C}_A$ and 
$\ket{\beta}_B =\sum_k b_k\ket{b_k} \in\B{C}_B$. 
The composed system is described by the state:
\begin{equation}\label{t3}
\ket{\alpha}\otimes\ket{\beta} = \sum_{i,k} a_i b_k \ket{a_i}\otimes\ket{b_k}.
\end{equation}
It is the particular case where the coefficients in Eq.\eqref{t2} are $c_{i,k} = a_i\cdot b_k$. 
If a composed system can be written as Eq.\eqref{t3} it is called a {\it product state}, and 
there is no correlations between $A$ and $B$. It can be checked easily via the mutual 
information of the state, which is clearly zero once that the von Neumman entropy of the 
pure state is zero. In the pure states subspace, the set of product states is a 
tiny set (\cite{popescu94,popescu95}, apud \cite{zyc98})

The concept of product state can be generalized for 
mixed state. Considering a composed system represented by the state 
$\rho_{AB}\in\M{D}(\B{C}_A\otimes\B{C}_B)$, it is called a {\it product state} if can 
be written as:
\begin{equation}
\rho_{AB} = \rho_A\otimes\rho_B,  
\end{equation} 
where $\rho_A\in\M{D}(\B{C}_A)$ and $\rho_B\in\M{D}(\B{C}_B)$ are the states of the system $A$ and $B$ 
respectively. The product state for mixed states is also no correlated, as its mutual information is zero. 
As the space of quantum states is a convex set, the convex combination of states 
will also be a quantum state. The convex combination of product states generalizes the notion of product states, that is named a 
{\it separable state} \cite{werner89}. 
\begin{defi}[Separable states]
Considering a composed system described by the state\\ 
$\sigma\in\M{D}(\B{C}_{A}\otimes\B{C}_{B})$, it is a  
{\it separable state} if and only if can be written as:
\begin{equation}
\sigma = \sum_{i,j}p_{i,j}\sigma_{i}^A\otimes\sigma_{j}^B,
\end{equation} 
where $\sigma_i^A\in\M{D}(\B{C}_A)$ and $\sigma_j^B\in\M{D}(\B{C}_B)$.
\end{defi}
The set of quantum channels that let separable states invariant is named 
{\it local operations and classical communication} (LOCC). The set of separable 
states form a subspace in the space of density matrices,  it can be denoted  
as $Sep(\B{C}_{AB})$.  The separable state can be easily extended to multipartite systems. Considering a 
$n$-partite system it is named $m$-separable if it can be decomposed in a convex 
combination of product states composed by $m$ parties.

\subsubsection{Entanglement Quantification}
A measure of entanglement for mixed state can be obtained from the 
quantification of entanglement for pure states. It is possible construct a measure of 
entanglement in this sense calculating the average of entanglement taken on 
pure states needed to form the state. The most famous measure which follow this 
idea is named {\it entanglement of formation}. The entanglement 
of formation is interpreted as the minimal pure-states entanglement required to build 
the mixed state \cite{bennettpra}.
\begin{defi}
Considering a quantum state $\rho\in\M{D}(\B{C}_{A}\otimes\B{C}_{B})$, 
the entanglement of formation is defined as:
\begin{equation}
E_f(\rho) = \min_{\xi_{\rho}}\sum_i p_i E(\ket{\psi_i}),
\end{equation}
where the optimization is performed over all ensembles $\xi_{\rho}=\{p_i,\ketbra{\psi_i}{\psi_i}\}_{i=1}^M$, 
such that $\rho = \sum_i p_i \ketbra{\psi_i}{\psi_i}$, $\sum_i p_i =1$ and $p_i\geq 0$. 
\end{defi}
The entanglement entropy $E(\ket{\psi_i})$ is defined as:
\begin{equation}
E(\ket{\psi_i})  = S(\Tr_B[\ketbra{\psi_i}{\psi_i}]),
\end{equation}
where $S(\Tr_B[\ketbra{\psi_i}{\psi_i}])$ is the von Neumann entropy of 
the reduced state of $\ket{\psi_i}$. 
The entanglement of formation is not easy to evaluate. Indeed the problem is 
related to find the minimal convex hull to form $\rho$ in function of a nonlinear function.
 For two qubits systems it can be calculated analytically \cite{wooters98}.

Quantum entanglement also enables an operational interpretation. This 
interpretation has two different ways: the resource required to 
construct a given quantum state, and the resource extracted from a quantum system. 
The resource here refers to the amount of copies of maximally mixture states. 
Then one can define the measure of this resources as a measure of entanglement in the limit of many copies. 

The number o copies $m$ of maximally entangled states required to construct 
$n$ copies of a given state $\rho$, by means of LOCC protocols, 
is named {\it entanglement cost} \cite{bennettpra}. 
The entanglement cost 
can be written as the regularized version of the entanglement of formation \cite{hayden01}. 
\begin{defi}[Entanglement cost]
The number of copies of the maximally entangled states required to build the 
state $\rho$ is given by:
\begin{equation}
E_C(\rho) = \lim_{n\rightarrow \infty} \frac{ E_{f}(\rho^{\otimes n})}{n},
\end{equation}
where $E_{f}(\rho^{\otimes n})$ is the entanglement of formation of 
the $n$ copies of $\rho$.
\end{defi} 

The number of copies $m$ of the maximally entangled state which 
can be extract from $n$ copies of a given state $\rho$,  by 
LOCC, is named {\it distillable entanglement} \cite{bennettpra}.
\begin{defi}[Distillable entanglement]
The distillable entanglement of a given state $\rho$ is defines as:
\begin{equation}
E_D(\rho) = \lim_{n\rightarrow \infty} \frac{m}{n},
\end{equation}
where $m$ is the number of maximally entangled states that can be extracted from 
$\rho$ in the limit of many copies. 
\end{defi}
The distillable entanglement is a very important operational measure of entanglement, 
because it quantifies how useful is a given quantum 
state, for the quantum information purpose. 

The operational meaning of the entanglement cost and the distillable entanglement 
compose the research theory o quantum entanglement. 
The entanglement 
cost and the distillable entanglement of a given state are not the same. Indeed the 
cost of entanglement is greater than the distillable entanglement. The point is: 
it is more expensive create a state $\rho$ with copies of maximally entangled state, than is possible  
to extract entanglement from $\rho$.  One example are the bound entangled state, even it is entangled 
it is not possible to extract any maximally mixture state, although it requires an amount of maximally entangled states to build it.

\subsection{Quantumness of Correlations}
\label{quantumness.disting}
This section presents a revision about some basic concepts of quantumness of correlations for distinguishable systems. The notion of 
classically correlated states and quantum discord are presented.  
\subsubsection{Classically   Correlated states}
Consider a flip coin game  
with two distinct events described by the 
states $\{\ketbra{0}{0},\ketbra{1}{1}\}$, each with the same probability 
$1/2$. It known that is possible to distinguish the faces of the coin, 
with a null probability of error. The probability of error to distinguish 
two events, or two probability distributions, depends on the trace distance 
of the probability vectors of the events:
\begin{equation}
P_E(\ketbra{0}{0},\ketbra{1}{1}) = \frac{1}{2} - \frac{1}{4}|| \ketbra{0}{0} - \ketbra{1}{1}||_1,  
\end{equation} 
as the states are orthogonal $|| \ketbra{0}{0} - \ketbra{1}{1}||_1=2$, therefore the probability 
of error $P_E(\ketbra{0}{0},\ketbra{1}{1})=0$, as one expected.  
Now suppose a quantum coin flip, described by the events 
$\{\ketbra{\phi}{\phi},\ketbra{\psi}{\psi}\}$, with equal probability 
$1/2$, where $\ket{\phi},\ket{\psi}\in \B{C}^2$. As an example, consider the 
states $\ket{\phi} = (\ket{0}+\ket{1})/\sqrt{2}$ and $\ket{\psi} = \ket{1}$. 
For this case the overlap is $\braket{\phi}{\psi}=1/\sqrt{2}$. The trace distance 
of these states is simply: \[|| \ketbra{\phi}{\phi} - \ketbra{\psi}{\psi}||_1=\sqrt{2}, \] 
then the probability of error to distinguish the events is not zero. Superposition of states in quantum mechanics creates events 
that cannot be perfectly distinguished. 
The distinguishability of quantum or classical events,  can be quantifier by the 
Jensen-Shannon divergence. 
For two probability distributions (or events)  
it is defined as the symmetric and smoothed version of 
the Shannon relative entropy, or in the quantum case the von Neumman 
relative entropy \cite{majteyjensen,jensenentropy}. 
\begin{defi}
The Jensen-Shannon divergence for 
two arbitrary events $\ket{\psi},\ket{\phi}$ is defined as:
\begin{equation}
J(\ket{\psi},\ket{\phi}) = \frac{1}{2}
S\left(\frac{\ketbra{\phi}{\phi}+\ketbra{\psi}{\psi}}{2}\Vert\ketbra{\phi}{\phi}\right)+
\frac{1}{2}
S\left(\frac{\ketbra{\phi}{\phi}+\ketbra{\psi}{\psi}}{2}\Vert\ketbra{\psi}{\psi}\right).
\end{equation}
\end{defi} 
For the classical coin flip game the Jensen-Shannon divergence will be just 
$J(\ket{0},\ket{1}) = 1$. On the other hand, for the quantum coin flip with 
states $\ket{\phi} = (\ket{0}+\ket{1})/\sqrt{2}$ and $\ket{\psi} = \ket{1}$, 
it will be $J(\ket{\phi},\ket{\psi}) = \sqrt{2}$. The Jensen-Shannon divergence 
is related to the Bures distance and induces a metric for pure quantum states related to the 
Fisher-Rao metric \cite{jensenmetric},  it is lager for more distinguishable events, and the largest distance characterizes complete 
distinguishable events.
The Jensen Shannon divergence  
for two arbitrary events 
$\ket{\psi},\ket{\phi}$ is related to the mutual information \cite{majteyjensen}:
\begin{equation}
J(\ket{\psi},\ket{\phi}) = I(R:E)_{\rho_{RE}},
\end{equation}
where $R$ represents a register, $E$ represents the events and 
$\rho_{RE}\in\M{D}(\B{C}_R\otimes\B{C}_E)$  characterizes the existence of two distinct events:
\begin{equation}
\rho_{RE} = \frac{1}{2}\ketbra{0}{0}_R\otimes\ketbra{\phi}{\phi}_E +
\frac{1}{2}\ketbra{1}{1}_R\otimes\ketbra{\psi}{\psi}_E. 
\end{equation}   
For the classical coin flip game it is:
$\rho^c_{RE} = \frac{1}{2}\ketbra{0}{0}\otimes\ketbra{0}{0} +
\frac{1}{2}\ketbra{1}{1}\otimes\ketbra{1}{1}$, with mutual information 
 $I(R:E)_{\rho^c_{RE}} = 1$. For the quantum coin   
the state will be  $\rho^q_{RE} = \frac{1}{2}\ketbra{0}{0}\otimes\ketbra{\phi}{\phi} +
\frac{1}{2}\ketbra{1}{1}\otimes\ketbra{\psi}{\psi}$, where for 
$\ket{\phi} = (\ket{0}+\ket{1})/\sqrt{2}$ and $\ket{\psi} = \ket{1}$, and the mutual 
information is $I(R:E)_{\rho^q_{RE}} = \sqrt{2}$. 
As the mutual information is a measure of correlations between two probability 
distributions, one realizes that there are more correlations between the 
register and the events when they are not completely distinguishable than 
in the completely distinguishable case. 
However 
two binary classical distributions cannot share more then one bit of information, in other 
words, their mutual information cannot be greater than one \cite{wilde}. As the correlations 
between the quantum coin events and the register are bigger than one, it means 
that there are correlations beyond the classical case. A quantum state is 
classically correlated if there exists a local projective measurement such that the 
state remains the same \cite{opp,henderson,zurek2001}. The state $\rho^c_{RE}$ is an example 
of \emph{ classical-classical state}. In general these states are defined as:
\begin{defi}[Classical-Classical States]\label{defi2}
Given a bipartite state $\rho_{AB}\in\M{D}(\B{C}_A\otimes\B{C}_B)$, it is strictly classically 
correlated (or classical-classical state) if there exists a local projective measurement  $\Pi_{AB}$  with elements 
$\{\Pi_l^A\otimes \Pi_k^B\}_{k,l}$ such that the post-measured state is equal 
to the input state:
\begin{equation}
\Pi(\rho_{AB}) = \sum_{k,l}\Pi_l^A\otimes \Pi_k^B \rho_{AB} \Pi_l^A\otimes \Pi_k^B = 
\rho_{AB},
\end{equation}
therefore $\rho_{AB} = \sum_{k,l}p_{k,l}\Pi_l^A\otimes \Pi_k^B$,  and $\Pi_x^Y = \ketbra{e_x}{e_x}_Y$ is a 
projetor in the ortonormal basis $\{\ket{e_x}_Y\}_x\in\M{H}_Y$.
\end{defi}
The state $\rho^q_{ER}$ is an example of a \emph{classical-quantum state}, because there exists a projective 
measurement, with elements $\{\ketbra{0}{0}, \ketbra{1}{1}\}$, 
over partition $E$ that keep the state unchanged. On the other hand, there is not a projective 
measurement over partition $R$ with this property. 
In general  a state $\rho_{AB}$ is classical-quantum if there exists a projective measurement
$\Pi_A$ with elements $\{ \Pi_k\}_k$ such that:
\begin{equation}\label{nondisc}
\Pi_A\otimes\id_B (\rho_{AB}) = \rho_{AB} = \sum_{k}p_k \Pi_k \otimes \rho_k. 
\end{equation}
As the set of classically correlated states is composed by block diagonal matrices  
it is not convex, once that 
combination of block diagonal matrices cannot be block diagonal. 
As the identity matrix is block diagonal, or just diagonal, 
this set is connected by the maximally mixed state, and it is a \emph{thin set }  \cite{ferraro}. 

\subsubsection{Quantum Discord} 
The amount of classical correlations in a quantum state is measured by the capacity 
to extract information locally \cite{fanchini2012}. As the measurement process is a classical 
statistical inference, classical correlations can be quantified  by the amount of correlations that are not destroyed by the local measurement. 

\begin{defi}
For a bipartite density matrix $\rho_{AB}\in\M{D}(\B{C}_A\otimes\B{C}_B)$, the classical 
correlations between $A$ and $B$ can be quantified by the amount of correlations 
that can be extracted via local measurements:
\begin{equation}\label{f1}
J(A:B)_{\rho_{AB}}= \max_{\id\otimes \M{B}\in\M{P}}{I(A:X)_{\id\otimes \M{B}(\rho_{AB})}}=\max_{\id\otimes \M{B}\in\M{P}} 
\left\{ S(\rho_A) - \sum_x p_x S(\rho_x^A) \right\},  
\end{equation}
where the optimization is taken over the set of local measurement maps 
$\id\otimes \M{B}\in\M{P}(\M{H}_{AB},\M{H}_{AX})$ and  
$\id\otimes \M{B}(\rho_{AB}) = \sum_x p_x \rho_{x}^{A}\otimes\ketbra{b_x}{b_x}$ 
is a quantum-classical state in the space $\M{B}(\M{H}_A\otimes\M{H}_X)$. 
\end{defi}

Originally H. Ollivier and W. Zurek \cite{zurek2001} have defined this expression 
restricting the optimization to projective measurements. Independently, 
L. Henderson and V. Vedral \cite{henderson} have defined the optimization of the 
classical correlations over general POVMs. 
As the mutual information quantifies the total amount of 
correlations in the state, it is possible to define a quantifier of quantum 
correlations as  
the difference between the total correlations in the system, 
quantified by mutual information, and  the classical correlations, 
measured by Eq.\eqref{f1}. This measure of quantumness of correlations 
is named \emph{quantum discord}:
\begin{defi}
The quantum discord $D(A:B)_{\rho_{AB}}$ of a state $\rho_{AB}$ is defined as:
\begin{equation}\label{f2}
D(A:B)_{\rho_{AB}} = I(A:B)_{\rho_{AB}} - J(A:B)_{\rho_{AB}},
\end{equation}
where $I(A:B)_{\rho_{AB}}$ is the von Neumann mutual information. 
\end{defi} 
Quantum discord quantifies the amount of information,  
that cannot be accessed via local measurements. Therefore 
it measures the quantumness shared between $A$ and $B$ 
that cannot be recovered via a classical statistical inference process. 
The optimization of quantum discord is a NP-hard problem \cite{huang13}. 
A general analytical solution for quantum discord is not known, nor a 
criterion for a giving POVM to be optimal. 
Nonetheless there are some analytic expressions for some specific 
states \cite{mazhar,girolami2011,lu}.  
It is natural a generalization of quantum discord 
for the case the measurement is performed locally on 
both subsystems. 
\begin{defi}
Given a bipartite state $\rho_{AB}\in\M{D}(\B{C}_A\otimes\B{C}_B)$ the quantum discord 
over measurements on both systems is:
\begin{equation}
D(A:B)_{\rho_{AB}} = \min_{\M{A}\otimes\M{B}\in\M{P}}
\left\{I(A:B)_{\rho_{AB}} - I(A:B)_{\M{A}\otimes\M{B}(\rho_{AB})} \right\},
\end{equation}
where $\M{A}\in \M{P}(\B{C}_{A},\B{C}_{Y})$ and $\M{B}\in \M{P}(\B{C}_{B},\B{C}_{X})$.
\end{defi}
This generalization of quantum discord was first discussed in \cite{piani08} in 
the context of the non-local-broadcast theorem. This definition is often named WPM-discord, 
because it was also studied by S. Wu \emph{et. al} \cite{wu09}. It was also 
studied restricting  to projective measurements by some 
authors \cite{girolami11,rulli11}.

\subsubsection{Relative entropy of quantumness and work deficit}\label{sec.req}
For a given dephasing channel $\Pi\in\M{P}(\B{C}^N)$, acting on any state $\rho\in\M{D}(\B{C}^N)$, the 
support of the dephased state contains the support of the input state: 
$\text{supp}(\rho)\subseteq \text{supp}(\Pi[\rho])$, therefore the measure of quantumness of 
correlations based on the relative entropy remains finite for every composed state \cite{zyc,wilde}.

Suppose Alice and Bob have a common composed system described by the state 
$\rho_{AB}\in\M{D}(\B{C}_A\otimes\B{C}_B)$ and they would like to extract 
work from this system. To accomplish their task, they can 
perform the {\it closed set of local operations and  classical communication} (CLOCC). 
This class of operations is composed by: {\it i)} addition of pure ancillas, 
{\it ii)} local unitary operations, and {\it iii)} local dephasing channels. 
Sending states through a local dephasing channel to represent the 
classical communication. If Alice and Bob are together in the same lab, 
they can extract work globally from the total system, then the total 
amount of information that Alice and Bob can extract from $\rho_{AB}$ together 
is defined as the {\it total work} \cite{opp}.
\begin{defi}
The work that can be extracted from a quantum system, described by the state $\rho\in\M{D}(\B{C}^N)$, 
is defined as the change in the entropy:
\begin{equation}\label{f11}
W_t(\rho) = \log_2{N} - S(\rho), 
\end{equation}  
$\log_2{N}$ is the entropy of the maximally mixed state, and $S(\rho)$ is the von Neumann entropy of the state. 
\end{defi} 
This function can be interpreted as a quantifier of information, such that if the 
state is a maximally mixed state no information can be extracted from it. Therefore 
if the state is a pure state we have the maximum amount of information \cite{opp,horodecki2005local}.  
The entropy function represents 
the amount of information that one {\it can get to know} about the system, therefore the function Eq.\eqref{f11} 
represents the amount of information that one {\it already knows}. 
On the other hand, if Alice and Bob cannot be in the same lab, the information  
that can be extracted from the total state is restricted to be locally accessed. In the same 
way it is possible to define the total information, named {\it local work}. 
Then Alice and Bob should perform CLOCC operation in order to obtain the maximal amount 
of local information \cite{horodecki2005local}:
\begin{equation}\label{f12}
W_l(\rho_{AB}) = \log_2{N} - \sup_{\Gamma\in CLOCC}S(\Gamma[\rho_{AB}]),
\end{equation}
where the state $\Gamma(\rho_{AB})$ is the state after the protocol. 
As CLOCC consist in send one part of the state in a dephasing channel, at the end of the protocol, the whole state is with the receiver: $\Gamma(\rho_{AB}) = \rho_{AA\apos}$.  
  
One can be interested in the amount of information that cannot be extracted locally by Alice and Bob. 
This function is named {\it work deficit} and it quantifies the amount of work that is not possible 
to extract locally \cite{opp}. 
\begin{defi}
Given a bipartite state $\rho_{AB}$, the information which two parts Alice and Bob cannot 
access, via CLOCC, is the work deficit:
\begin{equation}
\Delta(\rho_{AB}) = W_t(\rho_{AB}) - W_l(\rho_{AB}).
\end{equation}
\end{defi}
From the definition of the total work and the local work we can define the work deficit as the diference of them:
\begin{equation}\label{f13}
\Delta(\rho_{AB}) = \inf_{\Gamma\in CLOCC}\left\{ S(\Gamma[\rho_{AB}]) - S(\rho_{AB})\right\}.
\end{equation}
Even though the total and the local work depend explicitly on the dimension of the 
system, the work deficit should not depend on the dimension of $\Gamma[\rho_{AB}]$. 
Adding local pure ancillas belongs to the CLOCC cannot 
change the amount of work deficit. 
The work deficit 
can quantify quantum correlations, then it must not change by the simple addition 
of a uncorrelated system \cite{horodecki2005local,modiprl2010}. 

In the asymptotic limit (the limit of many copies) the work deficit quantifies the 
amount of pure states that can be extracted locally \cite{local,devetakdistillation}. 
However as a resource cannot be created freely, the addition of pure local ancillas 
is not allowed, then it is replaced by the addition of maximally mixed states. 
The set of operations that contains: {\it i)} addition of maximally mixture states, 
{\it ii)} local unitary operations, {\it iii)} local dephasing channels, is named 
{\it noise local operations and classical communication} (NLOCC) \cite{local}. 
The extraction of local pure states is a protocol, whose goal is to extract resource (coherence).  
The set of available operations 
are NLOCC operations, and the set of free resource states is composed only by 
the maximally mixture state. It is the only state without local purity \cite{michalopp2013}. 
It remains an open question if the CLOCC class and the NLOCC class are equivalent classes \cite{horodecki2005local}. 

In the limit of one copy, the work deficit can quantify quantum correlations present in 
a given composed system \cite{oppenheim2002new}. 
The scenario where Alice and Bob can perform many steps of 
classical communication one each other is named {\it two way}, and the work deficit is named 
{\it two-way work deficit}. In this case they can perform measurements and communicate 
 in each step of the protocol.
Mathematically the two-way work deficit does not have  
a closed expression \cite{horodecki2005local}. 
As discussed above, it is possible to activate  
quantum correlations performing operations on the measured system.  
Therefore this many steps scenario cannot quantify quantum correlations.  
Because if Alice and 
Bob can implement a sequence of non commuting dephasing channels, the only invariante state is the maximally mixed state. 
In this way, it is necessary a one round description, where 
 Alice and Bob can communicate in the end of the protocol. 
Following this idea, it is possible to define  the 
{\it one way work deficit}, which just one side can communicate. If Bob 
communicates to Alice, the state 
created in the end of the protocol is a quantum-classical state (or a classical-quantum 
state if Alice communicates in the end of the protocol).

\begin{defi}[one way work deficit]
Given a bipartite state $\rho_{AB}$, the work deficit with just one side communication 
is named {\it one way work deficit} \cite{opp}:
\begin{equation}
\Delta^{\rightarrow}(\rho_{AB}) = \min_{\Pi_{B}\in\M{P}}\left\{ S(\id_A\otimes\Pi_B[\rho_{AB}]) 
- S(\rho_{AB})  \right\},
\end{equation}
where $\Pi_B\in\M{P}(\B{C}_B)$ is a local dephasing on subsystem $B$.
The notation $\Delta^{\rightarrow}(\rho_{AB})$  means that 
the communication is from $A$ to $B$, and $\Delta^{\leftarrow}(\rho_{AB})$  in the opposite direction.
\end{defi}

Another definition for the work deficit is defined when both Alice and Bob communicate in the end of 
the protocol, this is named {\it zero way work deficit}. The state created in the end of 
the protocol is a classical-classical state.

\begin{defi}[zero way work deficit]
Given a bipartite state $\rho_{AB}$, the work deficit with no communication until the end of the protocol  
is named {\it zero work deficit} \cite{opp}:
\begin{equation}
\Delta^{\emptyset}(\rho_{AB}) = \min_{\Pi_A\otimes\Pi_{B}\in\M{P}}\left\{ S(\Pi_A\otimes\Pi_B[\rho_{AB}]) 
- S(\rho_{AB})  \right\},
\end{equation}
where $\Pi_A\otimes\Pi_B\in\M{P}(\B{C}_A\otimes \B{C}_B)$ is a local dephasing on subsystem $A$ and $B$. 
\end{defi}

In analogy with the work deficit, K. Modi {\it et. al} proposed  a measure of quantumness of correlation defined 
as the relative entropy of the state and the set of classical correlated states \cite{modiprl2010}. 
This measure is named {\it relative entropy of quantumness}. 
\begin{defi}[Relative entropy of quantumness]
The relative entropy of quantumness $D(\rho_{AB})_{QC}$ 
for a given state $\rho_{AB}$ is defined as the minimum 
relative entropy over the set of quantum-classical states \cite{modiprl2010}:
\begin{equation}\label{f10}
D(\rho_{AB})_{QC} = \min_{\xi_{AB}\in\Omega_{QC}}S(\rho_{AB}\Vert \xi_{AB}),
\end{equation}
where $\Omega_{QC}$ is the set of quantum-classical states.
\end{defi}
The relative entropy of quantumness for classical-classical states 
is denoted as $D(\rho_{AB})_{CC}$. It is analogous to Eq.\eqref{f10} 
when the optimization is taken over the set of classical-classical states $\Omega_{CC}$:
\begin{equation}\label{f111}
D(\rho_{AB})_{CC} = \min_{\xi_{AB}\in\Omega_{CC}}S(\rho_{AB}\Vert \xi_{AB}).
\end{equation}

As discussed previously, in the limit of one copy, the one way and the zero way 
work deficits quantify quantumness of correlations of the system.    
It is possible to obtain the equivalence between one way work deficit and relative entropy of quantumness.  

\begin{teo}\label{teo2}
The 1-way work deficit is equal to the relative entropy of quantumness 
for quantum-classical states \cite{horodecki2005local,modiprl2010}:
\begin{equation}
D(\rho_{AB})_{QC} = \Delta^{\rightarrow}(\rho_{AB}),
\end{equation}
\end{teo}
The same equivalence holds for zero way work deficit and the relative entropy of quantumness 
of classical-classical states:
\begin{equation}
D(\rho_{AB})_{CC} = \Delta^{\emptyset}(\rho_{AB}).
\end{equation}

The one way and zero way work deficits quantify quantumness correlations beyond the quantum 
entanglement, therefore we should be able to compare these two classes of quantum correlations. 
For the relative entropy this comparison is natural of the fact that CLOCC is a sub class of 
LOCC operations, which naturally implies that \cite{opp}: 
\begin{equation}
\Delta(\rho) \geq E_r(\rho), 
\end{equation}
where $\Delta(\rho)$ is the work deficit and  $E_r(\rho)$ is the relative entropy of entanglement. 
The equality is attached for bipartite pure states: $\ket{\psi}_{AB}\in\B{C}_A\otimes\B{C}_B$:
\begin{equation}
\Delta(\Psi_{AB}) =  E_r(\Psi_{AB}) = S(\rho_A), 
\end{equation}
where $\Psi_{AB} = \ketbra{\psi}{\psi}_{AB}$.
An interesting corollary of this proposition is that the quantum discord is equal 
to the work deficit for pure states, because it is also equal to the entropy of entanglement 
for pure states.

In this section it was introduced the concept of local disturbance, by the definition of the work deficit. That is the smallest relative entropy between the state  
and its local disturbed version (obtained performing a local dephasing channel on the state). Indeed there are many other local disturbance quantumness of 
correlations quantifiers, that can be obtained defining a quantum state discrimination measure, that can be for example Bures distance \cite{spehner13}, Schatten - p norm \cite{debarbapra}, trace distance \cite{dreply13}, Hilbert - Schmidt distance \cite{luo2010,piani12p}.  

\section{Monogamy relation: entanglement, classical correlations and quantumness of correlations}
\label{seckw}

Given a bipartite system $\rho_{AB}\in\mathcal{D}(\mathbb{C}_A\otimes \mathbb{C}_B)$, 
then it is possible to purify the state in a larger space $\mathbb{C}_{ABE}$, of the dimension: 
$\text{dim}(\mathbb{C}_{ABE})=\text{dim}(A)\cdot\text{dim}(B)\cdot \text{rank}(\rho_{AB})$. 
The purification process creates quantum correlations between the 
system $AB$ and the purification system $E$, unless the state is already 
pure. 
Intrinsically there is 
a restriction in the amount of correlations that can be shared by the systems.
This balance between the correlations for tripartite states 
can be understood by the Koashi-Winter relation. 

Given the definition of the classical correlations for a bipartite state $\rho_{AB}$:
\begin{equation}
J(A:B)_{\rho_{AB}} = \max_{\id\otimes\M \in \M{P}} I(A:X)_{\id\otimes\M(\rho_{AB})},
\end{equation}
where $I(A:X)_{\id\otimes\M(\rho_{AB})}$ is the mutual information of the post-measured 
state $\id\otimes\M(\rho_{AB})$, and the optimization is taken over all local POVM measurement maps 
$\M \in \M{P}(\B{C}_B,\B{C}_X)$.

Given also the definition of the entanglement of formation of a bipartite state $\rho_{AB}$:
\begin{equation}
E_f(\rho_{AB}) = \min_{\xi_{\rho} = \{p_i,\ketbra{\psi_i}{\psi_i}\}_i} \sum_i p_i E(\ket{\psi_i}),
\end{equation}
where the optimization is taken over all possible convex hull defined by the ensemble 
$\xi =\{p_i,\ketbra{\psi_i}{\psi_i}\}_i$, such that $\rho_{AB} = \sum_ip_i,\ketbra{\psi_i}{\psi_i}$, 
and $E(\ket{\psi_i})$ is the entropy of entanglement of $\ket{\psi_i}$. 

\begin{teo}[Koashi-Winter relation \cite{kw04}]
Considering $\rho_{ABE}\in\mathcal{D}(\mathbb{C}_A\otimes \mathbb{C}_B \otimes \mathbb{C}_E)$ a pure 
state then:
\begin{equation}\label{kw}
J(A:E)_{\rho_{AE}} = S(\rho_{A}) - E_f(\rho_{AB}),
\end{equation}
where $\rho_X = Tr_{Y}[\rho_{YX}]$.
\end{teo}
\begin{proof}
Suppose $\rho_{AB} = \sum_i p_i \ketbra{\psi_i}{\psi_i}$ is the optimum convex combination, such that 
$ E_f(\rho_{AB}) =\sum_i p_i S(Tr_B[\ketbra{\psi_i}{\psi_i}])$. 
The classical correlations in system $AE$ 
relates this decomposition with a measurement on the subsystem $E$.     
Therefore, there exists a measurement $\{ M_j^{E} \}$ on system $E$ 
such that $\rho'_{ABE} = \sum_j Tr_E[\rho_{ABE}(\id_{AB}\otimes M_j^E)]\otimes \ketbra{e_j}{e_j}_E$ 
and $Tr_{E}[\rho'_{ABE}] =\sum_i p_i \ketbra{\psi_i}{\psi_i} $. Tracing over subsystem $B$, then 
the post-measurement state will be: 
\begin{equation}
\rho'_{AE} =  \sum_j p_j Tr_B[ \ketbra{\psi_j}{\psi_j}] \otimes \ketbra{e_j}{e_j}, 
\end{equation}  
in this way, the mutual information of the post-measurement state:
\begin{eqnarray}
I(A:E)_{\rho'_{AE}} &=& S(\rho_A)+S(\rho'_E)-S(\rho'_{AE}),\\
&=& S(\rho_A)+ H(E) - H(E) - \sum_i p_i S(Tr_A[\ketbra{\psi_j}{\psi_j}]),\\
&=& S(\rho_A)- \sum_i p_i S(Tr_A[\ketbra{\psi_j}{\psi_j}]),\\
&=& S(\rho_A)- E_f(\rho_{AB}) ,
\end{eqnarray}
it was used the property of the Shannon entropy for a block diagonal state, where \\ 
$Tr_B[\ketbra{\psi_j}{\psi_j}] = Tr_A[\ketbra{\psi_j}{\psi_j}]$ and 
$ E_f(\rho_{AB}) =\sum_i p_i S(Tr_B[\ketbra{\psi_i}{\psi_i}])$. 
By definition\\ $J(A:E)_{\rho_{AE}}\geq I(A:E)_{\rho'_{AE}}$, then
\begin{equation}\label{e3}
J(A:E)_{\rho_{AE}}\geq S(\rho_A) - E_f(\rho_{AB}). 
\end{equation}
Now it is proved the converse inequality. Given $\rho_{AE}$, there exists a POVM 
$\mathcal{M}\in \mathcal{P}(\mathbb{C}_{E},\mathbb{C}_{E'})$ 
with  rank-1 elements $\{ M_l \}$, such that $Tr_E[M_l \rho_{AE}] = q_l \rho_l^A$ 
that optimizes the classical correlations $J(\rho_{AE}) = S(\rho_A) - \sum_l q_l S(\rho_l^A)$. 
As the elements of the POVM are rank-1, $M_l = \ketbra{\mu_l}{\mu_l}$, and 
the state $\rho_{ABE}$ is pure, the state after local measurement on $E$ 
will be described by an ensemble of pure states:
\begin{equation}
\rho'_{ABE} =\sum_l Tr_E[\rho_{ABE}(\id_{AB}\otimes \ketbra{\mu_l}{\mu_l})]\otimes \ketbra{e_l}{e_l}
= \sum q_l \ketbra{\phi_l}{\phi_l}\otimes \ketbra{e_l}{e_l}. 
\end{equation} 
Once that $\rho_{ABE}=\ketbra{\kappa}{\kappa}_{ABE}$, and the pure 
state can be written in the bipartite Schmidt decomposition 
$\ket{\kappa} = \sum_n c_n \ket{n}_{AB}\otimes\ket{n}_{E}$, if $\braket{\mu_l}{n} = r_{ln}$, 
it is easy to see that:
\begin{equation}
Tr_E[\rho_{ABE}(\id_{AB}\otimes \ketbra{\mu_l}{\mu_l})] = 
\sum_{ij}c_ir_{li} c_j r_{lj}^{*} \ketbra{i}{j}_{AB} = \left(\sum_{i}c_ir_{li}\ket{i}_{AB}\right)
\left(\sum_{j}c_jr_{lj}^{*}\bra{j}_{AB}\right) = q_l \ketbra{\phi_l}{\phi_l}.
\end{equation}
Calculating the mutual information of $\rho'_{AE} = Tr_B[\rho'_{ABE}]$:
\begin{equation}\label{e1}
I(A:E)_{\rho'_{AE}} = S(\rho_A) - \sum_l q_l S(Tr_B[\ketbra{\phi_l}{\phi_l}]),
\end{equation}
as the POVM $\M{M}$ is the optimal measurement in the calculation of the classical correlations 
it implies $I(A:E)_{\rho'_{AE}} = J(A:E)_{\rho_{AE}}$. 
By the definition of the entanglement of formation: 
$E_f(\rho_{AB}) \leq \sum_l q_l S(Tr_B[\ketbra{\phi_l}{\phi_l}])$ for any 
decomposition $\{p_l, \ketbra{\phi_l}{\phi_l} \}$. Substituting the mutual information 
on Eq.\eqref{e1}:
\begin{equation}\label{e2}
J(A:E)_{\rho_{AE}} \leq S(\rho_A) - E_f(\rho_{AB}).
\end{equation}
Given Eq.\eqref{e3} and Eq.\eqref{e2} it proves the theorem.
\end{proof}
The Koashi-Winter equation quantifies the amount of entanglement among $A$ and  
$B$, considering that the former is classically correlated with another system $C$. This property 
is interesting once that it is related to the monogamy of entanglement 
\cite{wooters00}, where the amount 
of entanglement shared by three parts is limited, and this limitation is given by the 
amount of classical correlations among the parties. This limitation holds for 
any tripartite state as stated in the following corollary:
\begin{cor}
For any tripartite state $\rho_{ABC}\in\mathcal{D}(\mathbb{C}_A\otimes \mathbb{C}_B \otimes \mathbb{C}_C)$, 
it follows:
\begin{equation}
E_f(\rho_{AB}) + J(A:C)_{\rho_{AC}} \leq S(\rho_A).
\end{equation}
The equality holds for $\rho_{ABC}$ pure.
\end{cor}
\begin{proof}
If $\rho_{ABC}$ is not a pure state, there exists a purification $\rho_{ABCE}$, 
then
$\mathbb{C}_A\otimes\mathbb{C}_B\otimes\mathbb{C}_{CE}$, then 
follows the last theorem:
\begin{equation}
J(A:CE)_{\rho_{ACE}}+ E_f(\rho_{AB})= S(\rho_A),
\end{equation}
therefore as the classical correlations are monotonic under local maps, 
then taking the trace over the system $E$ we have   
$J(A:CE)_{\rho_{ACE}}\geq J(A:C)_{\rho_{AC}}$. 
\end{proof}
As the Shannon entropy of $\rho_A$ represents the effective size of $A$ in qubits \cite{schumacher95}, 
this size can be approached as the capacity of the system $A$ makes correlations with other systems 
$B$ and $C$ \cite{kw04}. In other words, this means that the existence of the quantum or classical 
correlations between  $A$ and another system $B$ is enough to restrict the amount of 
quantum or classical correlations which $A$ can make with a third system $C$.

Summing the mutual information $I(A: E)_{\rho_{AE}}$ on both sides of the Koashi-Winter 
relation, Eq.\eqref{kw}, it is possible to obtain a monogamy expression for the entanglement 
of formation of the state $\rho_{AB}$ in function of the quantum discord \cite{fanchini11}: 
\begin{equation} \label{e4}
D(A:E)_{\rho_{AE}} =  E_f(\rho_{AB}) - S(A|E)_{\rho_{AE}},
\end{equation}
where $D(A:E)_{\rho_{AE}}$ is the quantum discord of the state $\rho_{AE}$ with 
local measurement on the subsystem $E$ and $S(A|E)_{\rho_{AE}} = S(AE) - S(E)$ 
is the conditional entropy. As the label in the states is arbitrary we can 
rewrite this expression changing the labels $E\rightarrow B$ and vice versa 
to obtain $D(A:B)_{\rho_{AB}} = S(A|B)_{\rho_{AB}} - E_f(\rho_{AE})$, 
taking the sum between this and Eq.\eqref{e4}:
\begin{equation}
D(A:E)_{\rho_{AE}} + D(A:B)_{\rho_{AB}} = E_f(\rho_{AE})+ E_f(\rho_{AB}),
\end{equation} 
as the total state is pure $S(A|E)_{\rho_{AE}} = -S(A|B)_{\rho_{AB}}$. 
This expression means that the sum of total amount of entanglement   
that $A$ shares with $B$ and $E$ is equal to the sum of the amount of quantum discord shared with $B$ and $E$ \cite{fanchini11}.

From Eq.\ref{e4} it is possible to calculate an interesting expression which 
relates the irreversibility of the entanglement distillation protocol  
and quantum discord \cite{cornelio11}. 
As discussed, the 
entanglement cost is larger than the distillable entanglement. 
Given the entanglement cost defined as the regularization of the 
entanglement of formation \cite{hayden01}:
\begin{defi}
For a mixed state $\rho_{AB}\in\M{D}(\mathbb{C}_A\otimes\mathbb{C}_B)$  
the regularization of the entanglement of formation $E_f(\rho_{AB})$
results in the entanglement cost: 
\begin{equation}
E_C(\rho_{AB}) = \lim_{n\rightarrow \infty}\frac{1}{n}E_{f}(\rho_{AB}^{\otimes n}).
\end{equation}  
\end{defi}
The Hashing inequality says that the distillable entanglement of $\rho_{AB}$ is lower bounded by the 
coherent information $I(A\rangle B)_{\rho_{AB}} = - S(A|B)$ \cite{devetak05}. As the 
coherent information can increase under LOCC it is possible to optimize it under LOCC  
attaining the distillable entanglement \cite{devetak05}. 
\begin{defi}
The regularized coherent information after optimization over LOCC for a mixed state $\rho_{AB}$ 
gives the distillable entanglement:
\begin{equation}
E_D(\rho_{AB}) = \lim_{n\rightarrow \infty}\frac{1}{n}I(A\rangle B)_{(V_n\otimes\id)\rho_{AB}^{\otimes n}},
\end{equation}
where $V_n\otimes\id$ acts locally on the $n$ copies of $\rho_{AB}$.
\end{defi} 
It is  also possible to define the regularized quantum discord:
\begin{defi}
The regularized quantum discord can be defined as the quantum discord of a state $\rho_{AB}$
in the limit of many copies:
\begin{equation}
D^{\infty}(A:B)_{\rho_{AB}} = 
\lim_{n\rightarrow \infty}\frac{1}{n}D(A:B)_{\rho_{AB}^{\otimes n}}.
\end{equation}
\end{defi}  
Therefore similarly to Eq.\eqref{e4} in the limit of many copies:
\begin{equation} 
D(A:E)_{\rho_{AE}^{\otimes n}} = E_f(\rho_{AB}^{\otimes n})-S(A|E)_{\rho_{AE}^{\otimes n}} ,
\end{equation}
taking the regularization we have:
\begin{equation} \label{e5}
D^{\infty}(A:E)_{\rho_{AE}} = E_C(\rho_{AB}) - S(A|E)_{\rho_{AE}},
\end{equation}
as the conditional entropy is additive $S(A|E)_{\rho_{AE}^{\otimes n}} = nS(A|E)_{\rho_{AE}}$.  
Therefore the follow theorem comes from \eqref{e4}. 
\begin{teo}[M. Cornelio {\it et al.} \cite{cornelio11}]
For every mixed entangled state $\rho_{AB}$, if
\begin{eqnarray}
E_D(\rho_{AB}) &=& \frac{1}{n}I(A\rangle B)_{(V_n\otimes\id)\rho_{AB}^{\otimes n}}\\
E_C(\rho_{AB}) &=& \frac{1}{k}E_{F}(\rho_{AB}^{\otimes n}),
\end{eqnarray}
for a finite number of $n$ and $k$, the entanglement is irreversible $E_C(\rho_{AB})>E_D(\rho_{AB})$.
\end{teo}
Taking the limit of many copies, the equation can be rewritten as:
\begin{equation}
D^{\infty}(A:E)_{\sigma_{AE}} = E_C(\sigma_{AB}) - E_D(\sigma_{AB}),
\end{equation}
where $\sigma_{AB} = (V_k\otimes\id)\rho_{AB} $ and $E_D(\sigma_{AB}) = kE_D(\rho_{AB})$. 
The quantum discord $D^{\infty}(A:E)_{\sigma_{AE}}$ in this context can 
be viewed as the minimal amount of entanglement lost in the distillation protocol, 
for states belonging to the class described in the theorem \cite{cornelio11}. 
This expression has an operational interpretation for 
quantum discord, where the quantum discord between the system and the purification system 
restricts the amount of e-bits lost in the distillation process. 
A similar interpretation can be obtained via the state merging protocol \cite{sm}, 
Alice (A), Bob (B) and the Environment (E) share a pure tripartite state $\rho_{ABE}$, 
she would like to send her state to Bob, keeping the coherence with  
the system $E$. They can perform this protocol consuming 
an amount of entanglement in the process, the amount of entanglement is 
the regularized quantum discord $D^{\infty}(A:E)_{\rho_{AE}}$ \cite{esm,dattasm}.   

In addition to the above relations, some upper and lower bounds between quantum discord 
and entanglement of formation have been calculated via the Koashi-Winter relation 
and the properties of entropy   \cite{yu,xi,xib,zhang}.
The Eq.\ref{e4} was also used to calculate the quantum discord and 
the entanglement of formation analytically for systems with  rank-$2$ and dimension 
$2\otimes n$ \cite{fanchini2012,lastra,cen}.
Experimental investigations of Eq.\eqref{e4} were performed in the characterization of the information flow between system and environment of a Non-Markovian process \cite{fffprl14}.  

\section{Activation protocol}
\label{secap}

Physically a measurement process  can be described as an interaction between the measurement 
apparatus and the system, followed by a projective measurement on the apparatus. 
Consider a state $\rho_{\M{S}}=\sum_k \lambda_k \ketbra{k}{k}
\in\M{D}(\mathbb{C}_{\M{S}})$. The input state, 
is described as $\rho_{\M{S}:\M{M}} = \rho_{\M{S}}\otimes \ketbra{0}{0}_{\M{M}}$, by coupling a pure ancilla, that represents de the measurement apparatus.
The interaction between the system and the ancillary state is  performed by a unitary evolution:  
$U_{\M{S}:\M{M}}\in\M{U}(\mathbb{C}_{\M{S}}\otimes\mathbb{C}_{\M{M}})$, such that 
$Tr_{\M{M}}[U_{\M{S}:\M{M}} \rho_{\M{S}:\M{M}} U_{\M{S}:\M{M}}^{\dagger}]
= \sum_l \Pi_l \rho_{\M{S}} \Pi_l^{\dagger}$. 
A unitary operation satisfing this condition is given by:
\begin{equation}
U_{\M{S}:\M{M}}\ket{k}_{\M{S}}\ket{0}_{\M{M}} = \ket{k}_{\M{S}}\ket{k}_{\M{M}},
\end{equation}
where $\{ \ket{k}\}$ is an orthonormal basis in $\mathbb{C}_{\M{S}}$. 
If the orthogonal basis $\{\ketbra{k}{k}\}$ is the canonical basis, this 
interaction is a Cnot gate \cite{chuang}. Therefore, after the interaction, 
the state will be:
\begin{equation}
\tilde{\rho}_{\M{S}:\M{M}}=U_{\M{S}:\M{M}}(\rho_{\M{S}:\M{M}})U_{\M{S}:\M{M}}^{\dagger} =  
\sum_k \lambda_k \ketbra{k}{k}_{\M{S}}\otimes \ketbra{k}{k}_{\M{M}}.
\end{equation}
The interaction between the system and the measurement apparatus results in a 
classically correlated state between the system and the apparatus. Hence 
performing a projective measurement on the state of the apparatus, the 
state of the system can be recovered. 

Suppose now that the state of the system is composed,  
for example a bipartite system 
$\B{C}_{\M{S}} = \B{C}_A\otimes\B{C}_B$.  The measurements are performed locally in 
each system, therefore the ancilla is also a bipartite system $\mathbb{C}_{\M{M}} = \mathbb{C}_{\M{M}_A}\otimes \mathbb{C}_{\M{M}_B}$. 
The unitary operator representing the interaction between the system and 
the measurement apparatus is $U_{\M{S}:\M{M}} = U_{A:\M{M}_A}\otimes U_{B:\M{M}_B}$. Then 
the post-measured state is:
\begin{equation}
\tilde{\rho}_{\M{S}}=Tr_{\M{M}}[U_{\M{S}:\M{M}}(\rho_{\M{S}}\otimes \ketbra{0}{0})U_{\M{S}:\M{M}}^{\dagger}]= 
\sum_{k,l} \Pi_{k}^{A}\otimes\Pi_l^{B} \rho_{AB} \Pi_{k}^{\dagger A}\otimes\Pi_l^{\dagger B}. 
\end{equation}
As aforementioned, the measurement process consists in interacting the system with an ancilla, 
that  represents the measurement apparatus, and then perform a projective measurement 
over the ancilla. 
However as the dimension of the ancilla 
is arbitrary, to represent a general measurement  (POVM),  
it is necessary to  
couple another ancilla with the same size of the state:
$\rho_{\M{S}':\M{M}}=\rho_{\M{S}}\otimes\ketbra{0}{0}_{\M{E}}\otimes \ketbra{0}{0}_{\M{M}}$, 
where $\ketbra{0}{0}_{\M{E}}$ is an ancillary state on space $\B{C}_{\M{E}}$. 
Then the interaction with the apparatus, given by a unitary evolution 
$U_{\M{S}':\M{M}}$, results in the post-measured state
\begin{equation}
\tilde{\rho}_{\M{S}} = Tr_{\M{M}}[U_{\M{S}':\M{M}} \rho_{\M{S}':\M{M}} U_{\M{S}':\M{M}}^{\dagger}] = 
\sum_{l} \Pi_l (\rho_{\M{S}}\otimes\ketbra{0}{0}_{\M{E}}) \Pi_l. 
\end{equation}
By the Naimark's theorem 
$Tr[\Pi_l (\rho_{\M{S}}\otimes \ketbra{0}{0}_{\M{E}})] = Tr[E_l \rho_{\M{S}}]$, where  
$E_l = (\id\otimes \bra{0}) \Pi_l (\id\otimes \ket{0})$ is a element of a POVM. 

A general a bipartite state can be written as $\rho = \sum_{i,j} \ketbra{i}{j}\otimes O_{i,j}$, 
where $O_{i,j}$ is an Hermitian operator with trace different from zero. 
Then if the measurement is performed only on the subsystem $A$, the 
state $\tilde{\rho}_{\M{S}:\M{M}}$ after the interaction with the measurement apparatus will be:
\begin{eqnarray}
\tilde{\rho}_{\M{S}:\M{M}} &=& U_{\M{S}:\M{M}}(\rho_{\M{S}:\M{M}})U_{\M{S}:\M{M}}^{\dagger} \\
&=& U_{A:\M{M}_A}\otimes\id_B\left( \sum_{i,j} \ketbra{i}{j}_{A} \otimes \ketbra{0}{0}_{\M{M}_A} \otimes O_{i,j}^{B}\right ) U_{A:\M{M}_A}^{\dagger}\otimes\id_B\\
&=& \sum_{i,j} \ketbra{i}{j}_{A} \otimes \ketbra{i}{j}_{\M{M}_A} \otimes O_{i,j}^{B}. 
\end{eqnarray}
Differently of the global measurement process, for local measurements, entanglement can be created during the measurement process. 
For example, if $O_{ij}=\frac{1}{2}\ketbra{i}{j}$, the system the interaction with the measurement apparatus creates a maximally entangle state. 
different from the case where the measurement is performed on the 
 A quantum state cannot create quantum entanglement with the measurement apparatus, if it is classically correlated. As proved in the following theorem. 
\begin{teo}[\cite{pianiprl,streltsov11}]
A state is classically correlated (has no quantumness of correlations),  
if and only if there exists an unitary operation such that the 
post interaction state is separable with respect to 
system and measurement apparatus.  
\end{teo}
\begin{proof}
The proof is performed for the general case, for  measurements on both systems.
\item [\bf If]: 
If the state is classically correlated:
\begin{equation}
\rho_{\M{S}} = \sum_{k,j} p_{k,j} \ketbra{a_k,b_j}{a_k,b_j}_{\M{S}},
\end{equation}
the state after the interaction with the measurement apparatus 
represented by the unitary operation $U_{A:\M{M}_A}\otimes U_{B:\M{M}_B}$ will be: 
\begin{equation}\label{max.cor}
\tilde{\rho}_{\M{S}:\M{M}} = \sum_{k,j} p_{k,j} \ketbra{a_k,b_j}{a_k,b_j}_{\M{S}}
\otimes\ketbra{a_k,b_j}{a_k,b_j}_{\M{M}}, 
\end{equation}
which is clearly separable.
\item [\bf Only if]: 
Given a general separable state between the system and the measurement apparatus:
\begin{equation}
\tilde{\rho}_{\M{S}:\M{M}} = \sum_{\alpha} p_{\alpha} 
\ketbra{\phi_{\alpha}}{\phi_{\alpha}}_{\M{S}}\otimes \ketbra{\psi_{\alpha}}{\psi_{\alpha}}_{\M{M}}, 
\end{equation}
and the fact that the interaction is unitary, there is a convex combination such that 
$\rho_{\M{S}}= \sum_{\alpha} p_{\alpha} \ketbra{\kappa_{\alpha}}{\kappa_{\alpha}}$, 
therefore the interaction must act in the following way:
\begin{equation}\label{e7}
U_{\M{S}:\M{M}} \ket{\kappa_{\alpha}}\ket{0} = \ket{\phi_{\alpha}}\ket{\psi_{\alpha}}. 
\end{equation}  
On the other hand, as the state $\rho_{\M{S}}$ is bipartite, the pure states 
$\{\ket{\kappa_{\alpha}}\}$ 
can be written in general as: 
$\ket{\kappa_{\alpha}}=\sum_{l,i} c_{l,i}^{\alpha} \ket{a_l^{\alpha}}\ket{b_i^{\alpha}}$, 
and after the interaction the states will be:
\begin{equation}\label{e8}
U_{\M{S}:\M{M}} \ket{\kappa_{\alpha}}\ket{0}  = 
\sum_{l,j} c_{l,j}^{\alpha} \ket{a_l^{\alpha},b_j^{\alpha}}_{\M{S}}\otimes \ket{a_l^{\alpha},b_j^{\alpha}}_{\M{M}}.
\end{equation}
As the state in Eq.\eqref{e8} must be separable, it implies that the 
coefficients must satisfy:
\begin{equation}
c_{i,j}^{\alpha} = c_{f(\alpha)}\delta_{i,j;f(\alpha)}\quad \text{and}\quad
|c_{f(\alpha)}|=1
\end{equation}
where $f(\alpha)\in \B{N}^2$. As $f(\alpha)$ are orthogonal it proves the theorem.\\ 
\end{proof}
If the state of the system has quantum correlations, the local measurement process creates entanglement between the system 
and the measurement apparatus, for a every unitary interaction.
Then it is possible to fix the base of the ancilla, and change the base of the system. Then rewriting the evolution as
$U_{\M{S}:\M{M}}= C_{\M{S}:\M{M}}(U_{\M{S}}\otimes \id_{\M{M}})$, where 
for bipartite systems $U_{\M{M}} = U_A\otimes U_B$ is a local unitary operation and 
$C_{\M{S}:\M{M}} = C_{A:\M{M}_A}\otimes C_{B:\M{M}_B}$ 
is a Cnot gate acting on the system as the control, and the apparatus as the target. 
It is possible to quantify the amount of  
quantum correlation in a given system starting on the amount of entanglement created with the measurement apparatus.
\begin{defi}[\cite{pianiprl,streltsov11}]
Each measure of entanglement used to quantify the entanglement 
between the system and the apparatus will result in a measure of 
quantumness of correlations.
\begin{equation}\label{e9}
Q_{E}(\rho_{\M{S}}) = \min_{U_{\M{S}}}E_{Q}(\rho_{\M{S}:\M{M}}). 
\end{equation} 
\end{defi}
Different entanglement
measures will lead, in principle, to different quantifiers 
for the quantumness of correlations. The only requirement 
is that the entanglement measure must be an entanglement  monotone  \cite{pianiprl,streltsov11,piani12}.  
Some quantifier of quantumness of correlations can be recovered with the activation protocol: 
the quantum discord \cite{streltsov11}, one-way work deficit \cite{streltsov11}, 
zero way work deficit \cite{pianiprl} and the geometrical measure 
of discord via trace norm \cite{pianiadesso}, are some examples.   
Taking the distillable entanglement in Eq.\eqref{e9} is quite simple to see that it results in zero way 
work deficit. As showed in Eq.\eqref{max.cor}, the interaction with the measurement apparatus results in the 
state 
\begin{equation}
\tilde{\rho}_{\M{S}:\M{M}} = \sum_{k,j} p_{k,j} \ketbra{a_k,b_j}{a_k,b_j}_{\M{S}}
\otimes\ketbra{a_k,b_j}{a_k,b_j}_{\M{M}}.
\end{equation}
That  is named {\it maximally correlated state}, and as showed in Ref.\cite{hiroshima}, the distillable 
entanglement of this state attach the Hashing inequality \cite{hashingineq}:
\begin{equation}\label{dist}
E_{D}(\tilde{\rho}_{\M{S}:\M{M}}) = - S(\M{S}|\M{M}),
\end{equation}
where $S(\M{S}|\M{M}) = S(\tilde{\rho}_{\M{S}}) - S(\tilde{\rho}_{\M{S}:\M{M}})$ 
is condicional entropy of $\tilde{\rho}_{\M{S}:\M{M}}$. On the other hand, the zero way work 
deficit for $\rho_{\M{S}}$ is:
\begin{equation}
\Delta^{\emptyset}(\rho_{\M{S}}) = \min_{\Pi_{\M{S}_A}\otimes\Pi_{\M{S}_B}\in\M{P}}\left\{ 
S(\Pi_{\M{S}_A}\otimes\Pi_{\M{S}_B}[\rho_{\M{S}}]) - S(\rho_{\M{S}})  \right\},
\end{equation}
where $\Pi_{\M{S}_A}\otimes\Pi_{\M{S}_B}\in\M{P}(\B{C}_{\M{S}_A}\otimes \B{C}_{\M{S}_B})$ is a local dephasing 
on subsystem $A$ and $B$. As $\tilde{\rho}_{\M{S}}$ is the measured state of the system and 
$\tilde{\rho}_{\M{S}:\M{M}}=U_{\M{S}:\M{M}} \rho_{\M{S}:\M{M}}U^{\dagger}_{\M{S}:\M{M}}$, then: 
\[ S(\tilde{\rho}_{\M{S}}) - S(\tilde{\rho}_{\M{S}:\M{M}}) = S(\Pi_{\M{S}_A}\otimes\Pi_{\M{S}_B}[\rho_{\M{S}}]) - S(\rho_{\M{S}}). \] 
Therefore: 
\begin{equation}\label{dist.zwwd}
\Delta^{\emptyset}(\rho_{\M{S}})=\min_{U_{\M{M}}}E_{D}(\tilde{\rho}_{\M{S}:\M{M}}).
\end{equation}
This equation means that the activation protocol creates distillable entanglement between the 
system and the measurement apparatus during a local measurement. In other words, quantumness 
of correlations of the system can be converted resource for quantum information protocol, and this 
conversion is ruled by the activation protocol. 

From the Eq.\eqref{e9}, it is possible to show that quantum entanglement is a lower 
bound for quantumness of correlations.
\begin{prop}[M. Piani and G. Adesso \cite{pianiadesso}]
For $\rho_{AB}\in\M{D}(\B{C}_A\otimes\B{C}_B)$: 
\begin{equation}
Q_E(\rho_{AB})\geq E_{Q}(\rho_{AB}),
\end{equation} 
where $Q_E$ and $E_Q$ are related by the Eq.\eqref{e9}.
\end{prop} 
To compare two measures of different quantities as 
quantumness of correlation and quantum entanglement, it is necessary a common rule.  
The activation protocol gives the rule to compare these two quantities and this 
rule says that the measures of quantumness of correlations and 
quantum entanglement must be related from Eq.\eqref{e9}. Entanglement is a lower bound for quantumness of correlations also in the 
geometrical approach \cite{debarbapra,dreply13}.

Activation protocol determines that a composed state is 
classically correlated if and only if it cannot create entanglement during the measurement process, for a given unitary interaction \cite{pianiprl,streltsov11,piani12}. 
This results provides an important tool for characterization of quantum correlations in identical particles systems (bosons and fermions), once that 
system and apparatus are distinguishable partitions, even if the particles in the system are identical. This approach was performed by Iemini et. al to 
prove how are the classically correlated states of bosons and fermions \cite{iemini14}. The activation protocol device also allows to determine the 
class of classically correlated states of the modes of a fermionic system, and its relation to the correlations of the fermions \cite{debarba17}.

The entanglement generation by means of quantumness of correlations, as stated by the activation protocol,  was experimentally evidenced using programmable quantum measurement \cite{adesso2014prl}. In the 
experiment setup the optimization on the unitary operations was performed by a set of programable quantum 
measurements in different local basis. 
As  quantumness of correlation can be generated by local operations \cite{zurek2001}, 
activation protocol was explored experimentally in the generation of 
distillable entanglement via local operations on the measured partition of the system \cite{brussprl15}.

\section{Conclusion}

This chapter leads to the fundamental aspects of quantum correlations: entanglement and 
quantumness of correlations. The purpose of the chapter was to demonstrate that 
quantumness of correlations plays an important role in entanglement resource theory, and by consequence 
in quantum information theory. It was presented that 
entanglement and quantumness of correlations connect each other in two different pictures.  
The relation derived by Koashi and Winter demonstrates the balance between quantumness of 
correlations and entanglement in the purification process.  
This balance leads to a formal proof for the irreversibility of the entanglement distillation protocol, 
in terms of quantumness of correlations. Indeed in this fashion quantumness of correlations revealed to 
play the main role in the state merging protocol, quantifying the amount of entanglement consumed during 
the protocol.   In the named {\it activation protocol}, the quantumness of correlations of a given composed 
 system can be converted into distillable entanglement with a measurement apparatus during the local measurement process. In resume the entanglement created by the interaction between the system and the measurement 
 apparatus is limited below by the amount of quantumness of correlations of the system.

\section{Acknowledgments}
This work is supported by INCT - Quantum Information.

\end{document}